\documentclass[11pt]{article}[fullpage,a4paper]

\usepackage{amsfonts,amsmath,amsthm,amssymb,braket,enumerate,mathtools,xcolor,float}
\usepackage[margin=2cm]{geometry}
\usepackage[backref,hidelinks]{hyperref}
\usepackage{comment}
\usepackage{todonotes}
\usepackage{authblk,nicefrac}
\usepackage{physics}
\usepackage{tikz, nicematrix}
\usepackage{indentfirst}
\usepackage{breqn}

\usepackage{tikz,ifdraft,wrapfig}
\usetikzlibrary{quantikz2}
\usetikzlibrary{arrows,positioning,matrix,backgrounds,calc,fit}
\tikzcdset{
every matrix/.style={ampersand replacement=\&}
}

\usepackage{thmtools}
\usepackage{thm-restate}

\makeatletter
\renewcommand*\env@matrix[1][*\c@MaxMatrixCols c]{%
  \hskip -\arraycolsep
  \let\@ifnextchar\new@ifnextchar
  \array{#1}}
\makeatother

\newtheorem{theorem}{Theorem}
\newtheorem{defn}{Definition}

\newtheorem{lemma}{Lemma}

\theoremstyle{remark}

\renewcommand{\title}[1]{\vbox{\center\LARGE{#1}}\vspace{5mm}}
\renewcommand{\author}[1]{\vbox{\center#1}\vspace{5mm}}
\newcommand{\address}[1]{\vbox{\center\em#1}}

\begin{document}

\begin{titlepage}
\begin{center}

\vskip 1cm

\setcounter{tocdepth}{2}

\title{Sublogarithmic Distillation in all Prime Dimensions using Punctured Reed-Muller Codes}

\author{Tanay Saha$^{1,2}$ and Shiroman Prakash$^{2}$}

\address{
${}^1$ Department of Mathematics, Simon Fraser University, Burnaby, B.C., Canada\\
${}^{2}$Department of Physics and Computer Science, Dayalbagh Educational Institute, Agra, India}

\date{} 
\end{center}

\vspace{0.3cm}
\begin{abstract}
Magic state distillation is a leading but costly approach to fault-tolerant quantum computation, and it is important to explore all possible ways of minimizing its overhead cost. The number of ancillae required to produce a magic state within a target error rate $\epsilon$ is $O(\log^{\gamma} (\epsilon^{-1}))$ where $\gamma$ is known as the yield parameter. Hastings and Haah derived a family of distillation protocols with sublogarithmic overhead (i.e., $\gamma < 1$) based on punctured Reed-Muller codes. Building on work by Campbell \textit{et al.} and Krishna-Tillich, which suggests that qudits of dimension $p>2$ can significantly reduce overhead, we generalize their construction to qudits of arbitrary prime dimension $p$. We find that, in an analytically tractable puncturing scheme, the number of qudits required to achieve sublogarithmic overhead decreases drastically as $p$ increases, and the asymptotic yield parameter approaches $\frac{1}{\ln p}$ as $p \to \infty$. We also perform a small computational search for optimal puncture locations, which results in several interesting triorthogonal codes, including a $[[519,106,5]]_5$ code with $\gamma=0.99$. 


\end{abstract}

\vfill

\end{titlepage}

\eject \tableofcontents

\clearpage

\newpage

\section{Introduction}
One of the most promising approaches to quantum fault tolerance is magic state distillation -- in which one first implements very high accuracy Clifford gates; these are used to  distill high accuracy `magic states' from a larger number of lower accuracy magic states \cite{MSD, knill2004faulttolerant}. These magic states are then used to perform high fidelity non-Clifford gates. Magic state distillation is estimated to be a major source of overhead\footnote{We should remark that, over the past decade, there have been several notable advances that may substantially reduce this overhead cost in practice (e.g., \cite{ litinski2019magic, chamberland2020very, bombin2024fault, gidney2024magicstatecultivationgrowing}.)}, and reducing this overhead is great theoretical and practical importance. While most work focuses on qubits, (e.g., \cite{Reichardt_2005, Reichardt_2009, Bravyi_2012,campbell2009structure, campbell2010bound, howard2016small, Haah_2017, Haah_2018, Hastings_2018, Nezami_2022}), there has been some significant work in magic state distillation with qudits of odd prime dimension $p$~\cite{PhysRevA.83.032310, ACB, CampbellAnwarBrowne, campbell2014enhanced, nature, DawkinsHoward, DPS2, Krishna_2019, Prakash2020contextual,2020golay, jain2020qutrit, prakash2024lowoverheadqutritmagic}. Some of these works, particularly, \cite{CampbellAnwarBrowne, campbellEnhanced, Krishna_2019, prakash2024lowoverheadqutritmagic}, suggest substantial reductions in overhead may be possible by increasing $p$. Here, we build on these works, by constructing new families of triorthogonal codes that allow for magic state distillation with sublogarithmic overhead, for qudits of arbitrary prime dimension $p$.  

Triorthogonal codes, first defined in \cite{Bravyi_2012}, are a special class of $[[n, k, d]]_p$ codes that admit a transversal non-Clifford gate from the third level of the Clifford hierarchy, known as a $T$ gate \cite{HowardVala}, and can be used to distill a corresponding magic state. Distilling via a triorthogonal code with parameters $[[n,k,d]]$, one can get $k$ output magic states with error rate $\epsilon_{out}=O(A_d \epsilon_{in}^d)$ using $n$ input states of error rate $\epsilon_{in}$, where $A_d$ is the number of logical-$Z$ operators in the code of weight $d$. Concatenating $z$ times, one obtains $k$ magic states with error rate $\epsilon_{\rm out} = O(A_d^{\frac{d^z-1}{d-1}} \epsilon_{in}^{d^z})$ from $\frac{n^z}{k^{z-1}}$ noisy magic states. Thus, the ratio of input to output magic states scales with the desired error rate, $\epsilon_{\rm out}$ as $O\pqty{\log^{\gamma} \left( \frac{1}{\epsilon_{out}} \right)}$ where the yield parameter, $\gamma = \frac{\log(n/k)}{\log(d)}$, characterizes the overhead \cite{MSD}. 

A Reed-Muller code was used in~\cite{MSD} to construct a $[[15,1,3]]_2$ code, now known to be triorthogonal, with a transversal non-Clifford gate that allowed for distillation with $\gamma=2.47$. \cite{Bravyi_2012} first formalized the concept of a triorthogonal code for qubits, and constructed a family of $[[3k+8, k, 2]]_2$ codes with $\gamma \to \log_2 3 \approx 1.6$.  Around the same time, the $[[15,1,3]]_2$ code of \cite{MSD} was generalized to qudits in \cite{CampbellAnwarBrowne}, who considered Reed-Muller codes of linear degree with a single puncture, and, subsequently, \cite{campbell2014enhanced}, who used Reed-Solomon codes of maximal degree with a single puncture to construct $[[n,1,d]]_p$ codes with $\gamma \to 1$, as $p \to \infty$. 

Magic state distillation protocols with $\gamma<1$ have sublogarithmic overhead. It was conjectured~\cite{Bravyi_2012} that $\gamma \geq 1$ for all protocols. This was disproved by \cite{Haah_2018, Hastings_2018}, who showed that punctured binary Reed-Muller codes give rise to triorthogonal codes with $\gamma<1$, with an explicit example with $\gamma=0.68$ -- the size of such codes is, however, extremely large. \cite{Krishna_2019} showed that, for qudits, Reed-Solomon codes of maximal degree with multiple punctures give rise to triorthogonal codes with $\gamma \to 0$ as $p \to \infty$ with much smaller block sizes, although the smallest $p$ for which $\gamma<1$ is $p=23$. 

More recently, families of binary triorthogonal codes with $\gamma \to 0$ have appeared \cite{wills2024constantoverheadmagicstatedistillation, golowich2024asymptotically}, perhaps making the question sublogrithmic distillation somewhat obsolete. Nevertheless, distillation protocols with sublogarithmic distillation are exceedingly rare, and seem to require inordinantly large block sizes, or qudit dimensionalities. It is therefore appears worthwhile to explore all possible constructions of such codes -- particularly, to help more precisely quantify the possible advantages of working with qudits over qubits. Moreover, there exists a very simple block construction of a family of small qutrit triorthogonal codes \cite{prakash2024lowoverheadqutritmagic}, for which $\gamma \to 1$. It appears to be a meaningful question to search for triorthogonal codes that outperform this simple family. 

Here, we present a natural extension of the work of \cite{CampbellAnwarBrowne, campbell2014enhanced, Krishna_2019} and generalize the constructions of \cite{Hastings_2018, Haah_2018} to qudits, to study qudit triorthogonal codes for all prime dimensions using punctured Reed-Muller codes of maximal degree. We construct a puncturing scheme in which the parameters of the resulting triorthogonal code are analytically computable, and show that these codes distill magic states with $\gamma < 1$ for qudits of any prime dimension. The size of the code, $n$, required to achieve $\gamma<1$ decreases drastically with $p$, as $p$ increases from $2$ to $23$. In the asymptotic limit of large $n$, the yield parameter of our codes vanishes as $\frac{1}{\ln p}$ as $p \to \infty$. We also perform a randomized computer search over puncture locations to construct qutrit and ququint punctured Reed-Muller codes with high yields -- this revealed a $[[519,106,5]]_5$ code with $\gamma=0.99$. To our knowledge, all  explicit constructions of binary triorthogonal codes giving rise to sublogarithmic distillation require block sizes several orders of magnitude greater.

Our paper is organized as follows. In section \ref{sec:triorthogonal-codes}, we review the construction of qudit triorthogonal codes from triorthogonal spaces. In section \ref{sec:reed-muller} we review Reed-Muller codes and discuss how they can be used to construct triorthogonal spaces, and then punctured to construct triorthogonal codes. In section \ref{sec:sublogarithmic}, we present a puncturing scheme generalizing that of \cite{Hastings_2018} that allows us to construct a family of $[[n,k,d]]_p$ triorthogonal codes, that give rise to distillation routines with sub-logarithmic yield when the codes are sufficiently large. An appendix contains more details about the asymptotic performance of this code. In section \ref{sec:other-schemes}, we summarize results from a computational search for other puncturing schemes that give rise to qutrit and ququint triorthogonal codes with better parameters. In section \ref{sec:conclusion}, we conclude with some discussion regarding the possible advantages of fault-tolerant quantum computing with qudits rather than qubits.

\section{Triorthogonal codes}
\label{sec:triorthogonal-codes}
First defined in~\cite{Bravyi_2012} and generalized to qudits with prime dimension in~\cite{Krishna_2019}, triorthogonal codes are CSS quantum error-correcting codes transversal in the $T$-gate, a diagonal gate from the third-level of the Clifford hierarchy, defined for qudits of arbitrary prime dimension in \cite{HowardVala, CampbellAnwarBrowne}. Triorthogonal codes are constructed from (classical) triorthogonal spaces, as we now review.

Let $p \in \mathbb{N}$ be prime and let $\mathbb{F}_p$ denote the finite field with $p$ elements. Define the $*$-product of two vectors $v=(v_1, v_2, \dots) \in \mathbb{F}_p^n$ and $w=(w_1, w_2, \dots) \in \mathbb{F}_p^n$ to be the vector $v*w=(v_1 w_1, v_2 w_2, \dots, v_n w_n)$.

\begin{defn}[Classical triorthogonal space]
    A linear space $\mathcal{C}$ over the field $\mathbb{F}_p$ is classical tri-orthogonal if:
    \begin{enumerate}
        \item $\forall x, y, z \in \mathcal{C}$, $\sum_i (x*y*z)_i=0 \; (mod \; p)$
        \item $\forall x, y \in \mathcal{C}$, $\sum_i (x*y)_i=0 \; (mod \; p)$.
    \end{enumerate}
\end{defn}
Note that, from its definition, any triorthogonal space is also self-orthogonal. 

Let $G$ be the generator matrix of a triorthogonal space. Let $G'$ be the generator matrix of a triorthogonal space punctured in $k$ locations, and let $G_0$ be the generator matrix of that space shortened in the same $k$ locations,
\begin{equation}
    G_{n_0 \times k_0} \xrightarrow[\text{$k$ columns}]{\text{Puncture @}} G' = \begin{pNiceMatrix}
        \dots & r_1 & \dots \\
        \dots & \vdots & \dots \\
        \dots & r_k & \dots \\
        \hline
        \dots & r_{k+1} & \dots \\
        \dots & \vdots & \dots
        \CodeAfter
        \tikz
        \draw [decorate, thick, decoration = brace]
        ([xshift=3mm]1-3.north east) to node [auto = left] {$G_1$} 
        ([xshift=3mm]3-3.south east) ;
        \tikz
        \draw [decorate, thick, decoration = brace]
        ([xshift=3mm]4-3.north east) to node [auto = left] {$G_0$} 
        ([xshift=3mm]5-3.south east) ;
        \end{pNiceMatrix}.
\end{equation}

Let $\mathcal G'$ and $\mathcal G_0$ be the linear classical codes generated by these generator matrices. Clearly, $\mathcal G_0 \subseteq \mathcal G'$ and $G'$ can be written as, $G'=\begin{pmatrix} G_1 \\ \hline G_0 \end{pmatrix}$. It is shown in \cite{Bravyi_2012} (and \cite{Krishna_2019} for qudits) that the $CSS(G_0 \to X, G'{}^{\perp} \to Z)$ is a CSS code with the desired transversal $T$-gate, known as a triorthogonal code. Assuming that the submatrix of $G$ that survives after puncturing has full rank, the length of the CSS code is $n_0-k$ and the dimension is $k$. For non-degenerate\footnote{If the code is degenerate, the distance of the code is the minimum Hamming weight of a vector in $\mathcal G_0/\mathcal G'{}^\perp$, which is greater than or equal to $d(\mathcal G_0^\perp)$. The codes studied in this paper are always non-degenerate.} $G_0$ and $G'$, the distance of this code is $\min(d(G_0^\perp), d(G'))=d(G_0^\perp)$. (This is because, by triorthogonality, $G_0 \subseteq G'{}^\perp$, and therefore $G' \subseteq G_0^\perp$.)

\section{Reed-Muller codes}
\label{sec:reed-muller}
Reed-Muller codes are constructed from the linear vector space of polynomials of total degree $r$ in $m$ variables over $\mathbb F_p$. They provide an elegant construction of triorthogonal codes, as observed in \cite{MSD, CampbellAnwarBrowne, Haah_2018, Hastings_2018}, as we now review.

\subsection{Triorthgonal spaces from Reed-Muller codes}

Let $RM_p(r, m)$ denote the classical Reed-Muller code over $\mathbb F_p$. A codeword of $RM_p(r, m)$ is a complete list of function values ${\rm ev}(f)=(f(\vec{v}) : \vec{v} \in \mathbb{F}_p^m) \in \mathbb{F}_p^{p^m}$ where $f$ is a polynomial of total degree at most $r$ in $m$ variables $(x_1, \ldots x_m)$ such that $x_i^p=x_i$. In other words, $RM_p(r, m) = \{ f : \deg(f) \leq r, f \in \mathbb{F}_p[x_1, x_2, \dots, x_m]/\langle \dots, x_i^p - x_i, \dots \rangle \}$.  We do not distinguish between a codeword and the polynomial that gives rise to that codeword in $RM_p(r, m)$. Note, also, that $RM_p(r, m)$ is invariant under permutations of coordinates that correspond to affine transformations of $\mathbb{F}^m$, because the transformation, $\vec{r} \to M \vec{r}+\vec{r}_0$ does not change the degree of $f(\vec{r})$.

The degree $r$ of a Reed-Muller code $RM_p(r, m)$ must satisfy \begin{equation}r \leq m(p-1), \label{max-degree-RM}\end{equation} because $x^p=x \mod p$.

It is a well-known fact that the dual of a Reed-Muller code is another Reed-Muller code.
\begin{lemma}
    $RM_p (r, m)^{\perp} = RM_p (\tilde{r}, m)$ where $\tilde{r}=m(p-1)-r-1$. \label{dual-thm}
\end{lemma}
This follows from the fact that 
\begin{equation}
        \sum_{x,y \in \mathbb{F}_p} x^{\alpha}y^\beta \neq 0
\end{equation} iff $\alpha = 0 \mod (p-1)$ and $\beta=0 \mod (p-1)$.

As an immediate corollary of the above lemma, we have the following.
\begin{theorem}
$RM_p (r, m)$ is a classical triorthogonal space if $3r< m(p-1)$. 
\end{theorem}
\begin{proof}
    Let $r_{\rm max}=\lfloor \frac{m(p-1)-1}{3} \rfloor$. In finite fields, $\sum_{x \in \mathbb{F}_p} x^r = 0 \; \forall r < p-1$. So, $|ev(f_1) * ev(f_2) * ev(f_3)|_M = |ev(f_1 f_2 f_3)|_M = 0$ if $r < r_{\rm max}$, using the summation relation on all $m$ variables.
\end{proof}

The distance of a classical Reed-Muller code $RM_p(r,m)$ is determined by polynomial in $m$ variables of degree $r$ with the most zeros. This is given by the Schwartz-Zippel lemma, \cite{DEMILLO1978193, schwartz1980, Zippel1979}, which we reproduce here,
\begin{theorem} The Schwartz-Zippel Lemma states that, given a multivariate polynomial $f(\vec{x})$ of degree $r$ in $m$ variables over $\mathbb{F}_p$, the total number of nonzero entries satisfies
\begin{equation}
    |\{\vec{x} | f(\vec{x}) \neq 0 \} | \geq p^{m-\lfloor\frac{r}{p-1}\rfloor} \left(1 - \frac{r \bmod (p-1)}{p} \right). 
\end{equation}
\end{theorem}
Here, $r \bmod (p-1)$ denotes the the remainder when $r$ is divided by $p-1$. For $p=2$, $r \bmod (2-1)=0$. This bound is tight, so the distance of a classical Reed-Muller code is, 
\begin{equation}
    d_{RM}(r,m) =p^{m-\lfloor\frac{r}{p-1}\rfloor} \left(1 - \frac{r \bmod (p-1)}{p} \right). 
\end{equation}

\subsection{Puncturing Reed-Muller codes}
We can form a triorthogonal code from this Reed-Muller space by puncturing it in a $k$ different locations, which we denote as $S=\{ \vec{x}_1, \ldots \vec{x}_k\}$. 

\begin{theorem}
    The distance of the quantum triorthogonal code constructed from puncturing a $RM_p(r,m)$ triorthogonal space is $d(PRM_p (\tilde{r}, m; S))$ where $\tilde{r}=m(p-1)-r-1$. \label{general-distance-theorem}
\end{theorem}

\begin{proof}
    $\mathcal G_0$ will then be a shortened Reed-Muller  code, denoted as $SRM_p(r,m; S)$, while $\mathcal G'$ will be a punctured Reed-Muller code, denoted as $PRM_p(r,m; S)$.
    Recall that the dual of a shortened code is a punctured code, so $\mathcal G_0^\perp=SRM_p(r,m; S)^\perp = PRM_p (\tilde{r}, m; S)$, and $d(PRM_p (\tilde{r}, m; S))$ assuming non-degeneracy.
     
    To see that degeneracy does not further increase the distance of the code, note that $\mathcal G'{}^\perp=PRM_p(r,m; S)^\perp = SRM_p (\tilde{r}, m; S)$. There always exists at least one codeword of minimum weight in $PRM_p (\tilde{r}, m; S)$ that is not be present in $SRM_p (\tilde{r}, m; S)$. A codeword of minimum weight $d$ in $PRM_p(\tilde{r},m; S)$ corresponds to a polynomial $f$ in $RM((\tilde{r},m)$. If $f$ vanishes on all puncture locations, one can act on it with an affine transformation to obtain another polynomial $f'$ that is non-zero on one of the puncture locations, which gives rise to codeword of weight $d' \leq d$ in $PRM_p(\tilde{r},m; S)$. Either $f$ or $f'$ is nonzero on one of the puncture locations, and is therefore absent from $SRM_p(\tilde{r},m; S)$.
\end{proof}  

Campbell \textit{et al.} \cite{CampbellAnwarBrowne} punctured $RM_p(1,m)$ once to construct triorthogonal codes for qudits. However, it is always advantageous to define a triorthogonal space using $RM_p(r,m)$, with $r$ taken to be as large as possible, subject to $3r<p(m-1)$. To see this, note that $\tilde{r}(r)$ from Theorem \ref{general-distance-theorem} is a monotonically decreasing function of $r$. For $\tilde{r}_1<\tilde{r}_2$, $PRM_p (\tilde{r}_1, m) \subset PRM_p (\tilde{r}_2, m)$, so $d(PRM_p (\tilde{r}_1, m)) \geq PRM_p(\tilde{r}_2, m)$. 

Reed-Muller codes of maximal degree, $r=r_{\rm max} = \lfloor \frac{p(m-1) -1}{3} \rfloor$ therefore provide a natural construction of triorthogonal spaces. However, one is then faced with the question: how does one choose the set of puncture locations? For fixed $k$, one may wish to choose $k$ puncture locations in order to maximize the distance $d$ of the resulting triorthogonal code. Moreover, one also wishes to vary the number of puncture locations $k$ in order to optimize $\gamma$. These appear to be non-trivial combinatorial questions.

\subsection{Codes with one or two punctures}
Let us first discuss the case of Reed-Muller codes with one or two punctures. Triorthogonal codes arising from Reed-Muller codes with one or two punctures do not give optimal yields. However, for a given $m$, they have the largest possible distance. The codes we construct in this subsection always outperform the punctured $RM_p(1,m)$ codes of \cite{CampbellAnwarBrowne}, but need not outperform the $RM_p(r,1)$ codes (i.e., Reed-Solomon codes) of \cite{campbellEnhanced}. 

We first have the following lemma. 
\begin{lemma} The classical distance of a punctured Reed-Muller code $RM_p (r, m)$ with one or two punctures is independent of the locations of the punctures. Assume $d_{RM}(r,m)\geq 2$. For one puncture, the distance is $d_{RM}(r,m)-1$, and for two punctures the distance is $d_{RM}(r,m)-2$. 
\end{lemma}
\begin{proof}
$RM_p(r,m)$ is invariant under permutations of coordinates that correspond to affine transformations of $\mathbb{F}_p^m$. We can therefore, without loss of generality choose the first puncture to be at the point $(0, \ldots, 0) \in \mathbb{F}_p^m$, and the second puncture to be at $(1, \ldots, 0)  \in \mathbb{F}_p^m$. 

Consider any minimum weight codeword in the unpunctured Reed-Muller code, which corresponds to a polynomial $f$ which is non-zero on at least 2 points in $\mathbb F_p^m$, by assumption. One can act on this polynomial with an affine transformation on $\mathbb F_p^m$ to obtain a minimal weight polynomial which is non-zero on both puncture locations. The distance of the code is therefore reduced by $1$ for each of the first two punctures.
\end{proof}
Note that, for $p>2$, there are several possibly-inequivalent choices for the location of the third and subsequent punctures, depending on whether the point $\vec{x}$ at which the third puncture is located is taken to be collinear with the first two punctures in $\mathbb F_p^m$. It is, therefore, also possible that, if the location is chosen carefully, the third puncture does not further reduce the distance of the resulting triorthogonal code.

The parameters of a triorthogonal code arising from puncturing $RM_p(r_{\rm max}, m)$ once or twice are $[[p^m -k, k, d_{RM}(\tilde{r}(r_{\rm max},m)]-k]_p$, for $k \leq 2$.

Parameters of the smallest ternary Reed-Muller codes with a single puncture are $[[8, 1, 2]]_3$, $[[26, 1, 2]]_3$, $[[80, 1, 5]]_3$, $[[242, 1, 8]]_3$, $[[728, 1, 8]]_3$, $[[2186, 1, 17]]_3$, and $[[6560, 1, 26]]_3$. The first two codes in the list, which have distance $2$, were constructed in \cite{CampbellAnwarBrowne}, but the remaining are new. In the limit $m \to \infty$, the yield parameter of such codes goes to $\gamma \to 3$; but, the yield parameter for these codes may be irrelevant in practice, as one may only need to do a single round of distillation with codes of reasonably large distance (e.g., \cite{fowler2013surface}). For $p=5$, small codes include those with parameters, $[[4, 1, 2]]_5$, $[[24, 1, 3]]_5$, $[[124, 1, 4]]_5$,  and $[[624, 1, 14]]_5$.  For $p=7$, we have $[[6, 1, 2]]_7$, $[[48, 1, 4]]_7$, $[[342, 1, 6]]_7$ and for $p=11$, we have codes $[[10, 1, 4]]_{11}$, $[[120, 1, 7]]_{11}$. The first of these are the Reed-Solomon codes first proposed in \cite{campbellEnhanced}.

As an example, consider the $[[80,1,5]]_3$. The noise reduction is plotted in Figure \ref{fig:80-1-5-threshold}. Its threshold of the code to depolarizing noise, defined via $\rho_M(\delta)=(1-\delta)\ket{M}\bra{M}+\delta \frac{1}{p}$, is $\delta*=0.16$. In general, for qudits of dimension $p$ there are $p-1$ distinct noise channels \cite{CampbellAnwarBrowne} -- we also checked that it distills all types of noise, not just depolarizing noise.

For more than two puncture locations, it appears that one must carry out a randomized or exhaustive search over puncture locations to obtain triorthogonal codes of high distance -- we carry out a modest search in section \ref{sec:other-schemes} similar to that provided for binary codes in \cite{Haah_2018}. Such searches are naturally limited to small codes over small primes $p$. In the next section, we present a puncturing scheme generalizing that used in \cite{Hastings_2018} to all prime dimensions. This scheme, while suboptimal, allows one to compute the distance analytically, for arbitrarily large codes.
\begin{figure}
    \centering
\includegraphics[width=0.5\linewidth]{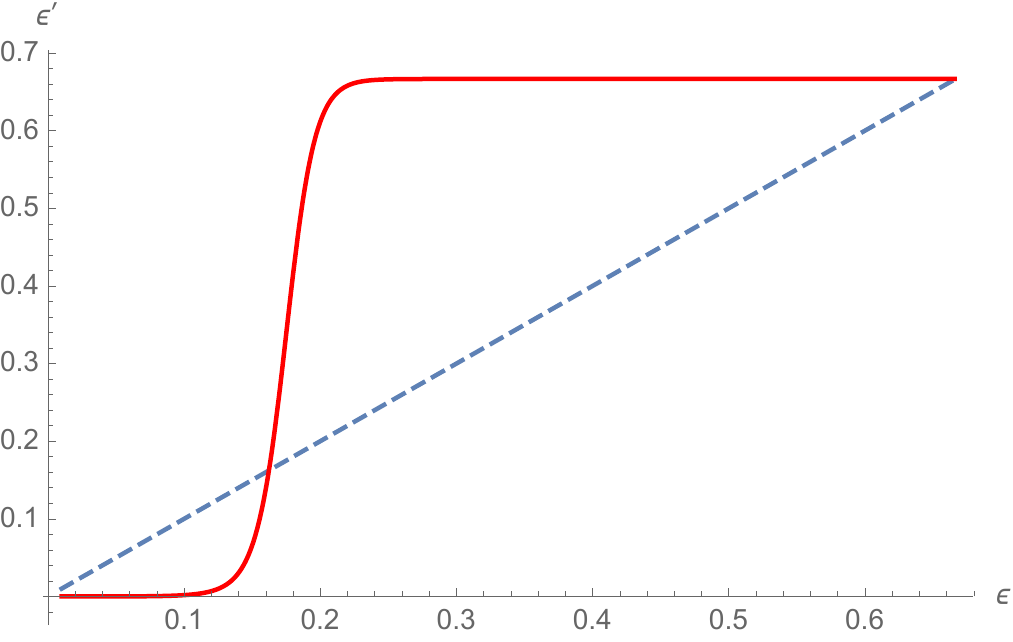}
    \caption{Distillation performance of the $[[80,1,5]]_3$ code, $\epsilon'(\epsilon)$ for depolarizing noise.}
    \label{fig:80-1-5-threshold}
\end{figure}

\section{Sublogarithmic distillation}
\label{sec:sublogarithmic}
In this section, we present an explicit scheme for puncturing triorthogonal spaces constructed from $RM_p(r,m)$ codes that give rise to triorthogonal codes that distill with sublogarithmic overhead for sufficiently large $m$ for all $p$. 

\subsection{Puncturing schemes based on weight functions}
Hastings and Haah \cite{Hastings_2018} devised a puncturing scheme for the binary Reed-Muller code $RM_2(r,m)$ which, although not optimal, allows one to derive explicit combinatorial expressions for the distance of the punctured Reed-Muller code for arbitrary $r$ and $m$. The key idea is puncturing in locations chosen according to the weight of the coordinate $\vec{x} \in \mathbb Z_2^m$ viewed as a binary codeword. One of the main results of the paper is that it is possible to generalize this idea to $RM_p(r,m)$. Unlike the binary case, however, there are many choices of the weight function possible.

Let $\vec{v} \in \mathbb F_p^m$ be a coordinate for $RM_p(r,m)$. A general weight function takes the form 
\begin{equation}
    |\vec{v}|_W = \sum_{i=1}^n W(v_i),
\end{equation} 
where $W: \mathbb F_p \to \mathbb Z$. The simplest example is the Hamming weight, $|\vec{v}|_H$, defined via 
\begin{equation}
    W_H(\alpha) = \begin{cases} 0 & \alpha =0 \\
    1 & \alpha \neq 0 
    \end{cases}.
\end{equation}
Another weight function that has appeared in the literature is, the Lee weight, $|\vec{v}|_L$ defined via,
\begin{equation}
    W_L(\alpha) = \begin{cases} \alpha & \alpha \leq (p-1)/2 \\
    p-\alpha & \alpha > (p-1)/2 
    \end{cases}.
\end{equation}
The Lee weight of a vector $\vec{v}$, is essentially the distance to the origin on a periodic $p\times p \times \ldots \times p$ lattice. 

We introduce a new distance function, which we find is more convenient for our purposes than the previous two weight functions, denoted as $|\vec{v}|_M$, which we call the \textit{Manhattan weight}, defined via,
\begin{equation}
    W_M(\alpha) = \alpha.
\end{equation}
$|\vec{v}|_M$ is the distance to the origin on a $p\times p \times \ldots \times p$ lattice that is not periodic using the Manhattan metric. In both the Lee and Manhattan weight functions, we implicitly convert $\alpha \in \mathbb F_p$ to $\alpha \in \{ 0, 1, 2, \ldots, p-1 \} \subset \mathbb Z$. 

Define a ``generalized binomial coefficient'', $C(m,k;W,p)$, via the generating function 
\begin{equation}
    \left(\sum_{j=0}^{p-1} x^{W(j)}\right)^m = \sum_{k=0}^{\infty} C(m,k;W,p) x^k.
\end{equation}
$C(m,k;W,p)$ counts how many words there are of weight $\tilde{W}(\vec{v})=k$ in the vector space $\mathbb F_p^m$. For Hamming weight, this takes the form
\begin{equation}
    \left(1+(p-1)x\right)^m = \sum_{j=0}^{m} C(m,j;W_H,p) x^j,
\end{equation}
and $C(m,j; W_H, p) = (p-1)^j \binom{m}{j}$.
For the Manhattan weight, 
\begin{equation}
    \left(1+x+x^2 + \ldots + x^{p-1}\right)^m = \sum_{j=0}^{m(p-1)} C(m,j;W_M,p) x^j.
\end{equation}

The puncturing scheme of \cite{Hastings_2018}, in this more general setting, is to puncture all coordinates with weight less than or equal to $w$, for some weight function. We denote the number of such coordinates by
\begin{equation}
C_\leq(m,k; W,p) = \sum_{j=0}^{k} C(m,j; W,p)
\end{equation}
and also for convenience, define,
\begin{equation}
    C_>(m,k;W,p) = \sum_{j=k+1}^{\infty} C(m,j; W,p)= p^m-C_<(m,k; W,p).
\end{equation}
For any weight function, this gives rise to a triorthogonal code with parameters\footnote{This formula, \eqref{triorthogonal-code-parameters}, assumes that the punctured generator matrix for the classical triorthogonal space continues to be of full-rank, which is necessarily true if the number of puncture locations is less than the minimum weight of a codeword in the unpunctured triorthogonal space, i.e., if $C_{\leq}(m,w; W,p) < d_{RM}(r,m)$.}
\begin{equation}
    [[C_>(m,w;W,p), C_\leq(m,w; W,p), d(m,r,w; W,p)]]_p, \label{triorthogonal-code-parameters}
\end{equation}
where $d(m,r,w; W,p)$ is the distance of a punctured Reed-Muller code, and depends non-trivially on both the choice of $W$, and the degree $r$ used to define the Reed-Muller code.

None of these puncturing schemes is optimal. Note, however, that $|\vec{v}|_M \geq |\vec{v}|_L \geq |\vec{v}|_H$, and therefore $C(m,k;W_M,p) \leq C(m,k;W_L,p) \leq C(m,k;W_L,p)$. Therefore, of the weight functions listed above, puncturing according to Manhattan weight will give us the most freedom in choosing the total number of puncture locations -- which can then be tuned to optimize $\gamma$. It turns out that $d(m,r,w; W_M,p)$ can be calculated in this case following the same strategy as \cite{Hastings_2018}. We therefore focus exclusively on this case.

\subsection{Combinatorial tools: $p$-nomial coefficients}
\label{sec:combinatorial-tools}
Let us introduce some combinatorial tools to aid in counting sets of specified Manhattan weight.

We introduce the notation $\binom{m}{s}_p=C(m,s;W_M,p)$, i.e.,  
\begin{equation}
    (1 + x + x^2 + \dots + x^{p-1})^m = \sum_{s=0}^{m(p-1)} \binom{m}{s}_p x^s.
\end{equation} 
We refer to $\binom{m}{j}_p$ as \textbf{$p$-nomial coefficents}.\footnote{ See \cite{weisstein_trinomial_2023} and references therein for discussion of these coefficients for the case $p=3$, which some authors call trinomial coefficients. $p$-nomial coefficients for general $p$ were discussed by Euler in  \cite{euler2005expansionpowerpolynomial1xx2x3x4etc}.}  Notice that when $p=2$, the $p-$nomial coefficients are the familiar binomial coefficients.

Many of the identities satisfied by binomial coefficients generalize immediately to $p-$nomial coefficients. We list some of these below.

Consider,
\begin{align}
        (1 + x +  \dots + x^{p-1})(1 + x +  \dots + x^{p-1})^m & = (1 + x  + \dots + x^{p-1})\sum_{s=0}^{m(p-1)} \binom{m}{s}_p x^s \\
        (1 + x  + \dots + x^{p-1})^{m+1} & = \sum_{s=0}^{(m+1)(p-1)} \binom{m+1}{s}_p x^s.
\end{align}
Comparing the coefficients, we get the following lemma.
\begin{lemma}   \label{eq:genPascal}
    Generalized Pascal's Rule:     \begin{equation}
        \binom{m+1}{s}_p = \sum_{i=0}^{p-1} \binom{m}{s-i}_p.
    \end{equation}
\end{lemma}

Using the multinomial expansion, the $p-$nomial coeffients can be expressed as a sum of multinomial coefficients:
\begin{lemma}
The coefficient of $x^s$ in $(1 + x + x^2 + \dots + x^{p-1})^m$ is
    \begin{equation}
        \binom{m}{s}_p = \sum_{k_1 + 2 k_2 + \dots + (p-1) k_{p-1} = s; k_i \geq 0} \frac{m!}{k_1! k_2! \dots k_{p-1}!(m-k_1-k_2-\dots-k_{p-1})!}.
    \end{equation}
\end{lemma}

In what follows, it is also convenient to introduce the following short-hand notation, following \cite{Hastings_2018},
\begin{defn}
    We define
    \begin{align}
        \binom{m}{>s}_p = \sum_{i>s} \binom{m}{i}_p \text{ and } \binom{m}{\leq s}_p = \sum_{i \leq s} \binom{m}{i}_p.
    \end{align}
\end{defn}
Note that $\binom{m}{>s}_p \geq 0$ $\binom{m}{\leq s}_p \geq 0$ $\binom{m}{s}_p \geq 0$, and are only defined for positive $m$, $s$.

It is straightforward to derive the following lemma.
\begin{lemma}
    \begin{equation}
        \binom{m}{>k}_p+\binom{m}{\leq k}_p=p^m.
    \end{equation}
\end{lemma}

In Appendix \ref{asymptotic-appendix}, we discuss the asymptotic form of the $p-$nomial coefficients, $\lim_{m \to \infty} \frac{1}{m} \log_p {\binom{m}{m \theta}_p}=H_p(\theta)$.  

\subsection{Distance of the punctured Reed-Muller code}

We now compute the minimum distance of the classical Reed-Muller code, punctured according to the Manhattan weight scheme.

Reed-Muller codes are defined as the evaluations of polynomials in $\mathbb{F}_p[x_1, x_2, \dots, x_m]/\langle x_i^p-x_i \rangle$ over the vectors $\vec{x} \in \mathbb{F}_p^m$. Puncturing this code is equivalent to removing a subset $S \subset \mathbb{F}_p^m$ of points from $\mathbb{F}_p^m$. For our construction, we choose this subset to be $S_p(w) = \{v \in \mathbb{F}_p^m : |v|_M \leq w \}$, and we call the punctured code $PRM_p(r, m, w)$. 

Each function $f \in \mathbb{F}_p[x_{1}, x_{2}, \dots, x_{m}] / \langle x_{i}^p - x_{i} \rangle$ with total degree less than $r$ corresponds to a codeword in $PRM_p(r,m,w)$. Let $|f|_{>w}$ denote the support of $f$ in $\mathbb{F}_p^m \setminus S_p(w)$. $|f|_{>w}$ is the Hamming weight of the codeword corresponding to $f$ in $PRM_p(r,m,w)$. The function with smallest non-zero $|f|_{>w}$ determines the minimum distance of the code. 

\begin{theorem}
    The distance of the classical punctured Reed-Muller code, $PRM_p(r, m, w)$, is \begin{equation}
    \begin{split}
      \Delta_p(m, r, w) & = \sum_{j=0}^{p-\beta-1} \binom{m-\alpha-1}{>w-j}_p  \\
      &  = \begin{cases}
        \binom{m-\alpha}{>w}_p & \beta=0 \\
        \binom{m-\alpha}{>w}_p - \sum_{j=p-\beta}^{p-1} \binom{m-\alpha -1}{>w-j}_p & \beta \neq 0, ~\alpha \leq m-1.
    \end{cases}
    \end{split}     
    \end{equation} 
    where $r=\alpha (p-1) + \beta \leq m(p-1)$ with $\beta \in \{ 0, 1, \dots, p-2 \}$, assuming that $\Delta_p(m,r,w) > 0$. \label{distance-theorem}
\end{theorem}
\begin{proof}
We will show that any nonzero function of total degree $r$ in $\mathbb F_p[x_1, \ldots, x_m]/\langle x_i^p-x_i\rangle$  has support  in $\mathbb F_p^m \setminus  S_p(w)$ greater than or equal to $\Delta_p(m,r,w)$, and that there is a function that saturates this inequality. 
If $\Delta_p(m,r,w)>0$, then $\Delta_p(m,r,w)$ is the minimum distance of the punctured Reed Muller code, as per the discussion above.

Consider the polynomial
\begin{equation}
    m_0(\vec x)=\prod_{i=1}^{\alpha}(1-x_i^{p-1}) \cdot \prod_{j=0}^{\beta-1} (p-1-i-x_{\alpha+1}) \in RM_p(r, m).
\end{equation}
This polynomial has degree $r$ and is thus a codeword in $PRM_p(r,m,w)$. 

Let us show that $|m_0(\vec{x})|_{>w}=\Delta_p(m,r,w)$, which is the number of points in $\mathbb{F}_p^m \setminus S_p(w)$ where $m_0(\vec{x}) \neq 0$. From the form of $m_0(\vec{x})$, we see that any such point must have $x_1 = x_2 = \dots = x_{\alpha}=0$ and $x_{\alpha+1} \in \{ 0, 1,  \ldots, p-\beta-1 \}$. Thus, the weight of this codeword is the number of solutions to $x_{\alpha+1} + x_{\alpha+2} + \dots + x_m > w$ where $x_i \in \{ 0, 1, \dots, p-1 \}$ for $i > \alpha + 1$. This is given by $\sum_{j=0}^{p-\beta-1} \binom{m-\alpha-1}{>w-j}_p=\Delta_p(m,r,w)$.

We will now prove that no function $f$ in $PRM_p(r,m,w)$ has $|f|_{>w}<\Delta_p(m,r,w)$. We will proceed by induction on the number of variables $m$, recalling that the degree must satisfy, $r=\alpha (p-1)+\beta\leq m(p-1)$.  The base case of $m=1$ (for which $r \leq p-1$) is trivial to show. Let $f^{(m)}$ be any polynomial in $m$ variables such that $\deg(f) \leq \min(r , m(p-1))$. We assume, for the inductive argument, that $|f^{(m)}(x_{1}, x_{2}, \dots, x_{m})|_{>w} \geq \Delta_{p}(m, r, w)$.  We now show that this assumption implies that any polynomial $f^{(m+1)}$ in $m+1$ variables, with $\deg(f^{(m+1)}) \leq \min(r, (m+1)(p-1))$, then  $|f^{(m+1)}|_{>w} \geq \Delta_{p}(m+1, r, w)$. 

Note that, if $r=(m+1)(p-1)$, the result is trivial, since $\Delta_p(m+1, (m+1)(p-1),w)=0$. So we assume $\deg(f^{(m+1)})=r<(m+1)(p-1)$ in what follows. (This does not, however, imply that $r<m(p-1)$.)
Any $f^{(m+1)}(x_{1}, x_{2}, \dots, x_{m+1})$ can be written as
\begin{equation}
    f^{(m+1)} = f^{(m)}_{0} + f^{(m)}_{1} x_{m+1} + f^{(m)}_{2} x_{m+1}^2 + \dots + f^{(m)}_{p-1} x_{m+1}^{p-1}
\end{equation}
where $\deg(f^{(m)}_{i}) \leq \min(r-i, m(p-1))$. Thus, summing over all values of $x_{m+1}$, we have,
\begin{equation}
\begin{split}
    |f^{(m+1)}|_{>w}  = & \left|f^{(m)}_{0}\right|_{>w} + \left|f^{(m)}_{0} + f^{(m)}_{1} + \dots + f^{(m)}_{p-1}\right|_{>w-1} \\ & + \left|f^{(m)}_{0} + 2 f^{(m)}_{1} + \dots + 2^{p-1}f^{(m)}_{p-1}\right|_{>w-2} + \dots \\ & + \left|f^{(m)}_{0} - f^{(m)}_{1} + f^{(m)}_{2} - \dots + f^{(m)}_{p-1}\right|_{>w-(p-1)}. \label{decomposition}
\end{split}
\end{equation}
We denote each of the terms in the above equation as,
\begin{equation}
f^{(m+1)}\big|_{x_{m+1}\to z}=f^{(m)}_0+f^{(m)}_1 z+ f^{(m)}_2 z^2 + \ldots f^{(m)}_{p-1}z^{p-1}. 
\end{equation}
(So, e.g., $f^{(m+1)}\big|_{x_{m+1}\to 0}=f^{(m)}_0$, and $f^{(m+1)}\big|_{x_{m+1}\to 2}=f^{(m)}_{0} + 2 f^{(m)}_{1} + \dots + 2^{p-1}f^{(m)}_{p-1}$.)

Let us first consider 
\begin{itemize}
 
\item Case 1:  $\forall z \in \mathbb F_p$, $f^{(m+1)}\big |_{x_{m+1}\to z}$ does not identically vanish.
    
Each of the terms $f^{(m+1)}\big |_{x_{m+1}\to z}$ in eq. \eqref{decomposition} has $m$ variables and  degree $r^* \leq \min(r,m(p-1))$. 

Then, by the induction hypothesis, 
\begin{align}
    |f_{0}|_{>w} &\geq \Delta_{p}(m, r^*, w) \geq \Delta_{p}(m, r, w) \\
    |f_{0} + f_{1} + \dots + f_{p-1}|_{>w-1} &\geq \Delta_{p}(m, r^*, w-1)  \geq \Delta_{p}(m, r, w-1) \\
    &\vdots \\
    |f_{0} - f_{1} + f_{2} - \dots + f_{p-1}|_{>w-(p-1)} &\geq \Delta_{p}(m, r^*, w-(p-1))  \geq \Delta_{p}(m, r, w-(p-1)).
\end{align}
Above, we used the fact that $\Delta_{p}(m, r, w)$ is a decreasing function of $r$, and $r^* \leq r$. 

Then, using the \hyperref[eq:genPascal]{generalized Pascal's rule}, we have
\begin{equation}
    |f|_{>w} \geq \sum_{i=0}^{p-1} \Delta_{p}(m, r, w-i) =\Delta_{p}(m+1, r, w).
\end{equation}

\item Case 2: More generally, if $f^{(m+1)}\big|_{x_{m+1}\to z_j} = 0$ for some $\{ z_j\}$, for $j=1, \ldots A$, then $$f^{(m+1)}=(x_{m+1}-z_1)(x_{m+1}-z_2)\cdots (x_{m+1}-z_A) h^{(m+1)},$$ with $h^{(m+1)} \neq 0$. Note that $h^{(m+1)}$ may still depend on $x_{m+1}$, but $\deg h^{(m+1)}=\min(r-A,(m+1)(p-1))=r-A$. We count the size of the support of $h^{(m+1)}$ as in Case 1, but we should not count any points for which $x_{m+1}=z_j$, for any $z_j$. Then, we have
\begin{equation}
\begin{split}
    |f^{(m+1)}|_{>w} & \geq \sum_{i=0 ,~ i \notin \{z_1, \ldots, z_A\} }^{p-1} \Delta_{p}(m,r-A, ~w-i)  \geq \sum_{i=0}^{p-1-A} \Delta_{p}(m, r-A, w-i).
\end{split}    
\end{equation}
Using Lemma \ref{distance-lemma-2},  in Appendix \ref{app:p-nomial}, we have
\begin{equation}
    \sum_{i=0}^{p-1-A} \Delta_{p}(m, r-A, w-i) \geq  \Delta_{p}(m+1, r, w).
\end{equation}
Therefore, again $|f|_{>w} \geq  \Delta_{p}(m+1, r', w)$ as required. 
\end{itemize}

\end{proof}

Note that this proof also implies that the generator matrix for $PRM_p(r,m,w)$ is full rank, if $\Delta_p(m,r,w) > 0$. In those cases where $\Delta_p(m,r,w)=0$, the punctured generator matrix is not of full rank, and this theorem does not determine the distance of the punctured code. 

\subsection{A family of quantum Reed-Muller codes}
\begin{theorem}
    Puncturing $RM_p(r,m)$ codes according to Manhattan weight gives rise to a family of qudit (with prime dimension $p$) triorthogonal codes with parameters $[[\binom{m}{>w}_p, \binom{m}{\leq w}_p, \Delta_p(m, \tilde{r}, w)]]_p$ with $\tilde{r}=m(p-1)-r-1$, $3 r < m(p-1)$ and $w < m (p-1) - r$, if $\Delta_p(m,r,w) > 0$. \label{code-parameters}
\end{theorem}
\begin{proof}
We have already shown that $n$ is given by $\binom{m}{>w}_p$. The number of puncture locations is $\binom{m}{\leq w}_p$; this is equal to $k$ if the punctured triorthogonal matrix is of full rank, which is guaranteed by Theorem \ref{distance-theorem}, if $\Delta_p(m,r,w) > 0$. From Theorem \ref{general-distance-theorem}, the distance of the quantum triorthogonal code is  $d(\mathcal G_0^\perp) = d(PRM_p(m(p-1)-r-1, m, w))$, which is given by Theorem \ref{distance-theorem}.
\end{proof}

Let us study the overhead of this code in the asymptotic limit of large $m$ when $r$ is maximal. To facilitate the analysis, we choose $m=3 \alpha$, where $\alpha$ is a large positive integer. Then, $r=\alpha (p-1) - 1$ is the largest value of $r$ such that $3r<m (p-1)$. For this choice of parameters, the distance of our quantum code simplifies to \begin{equation}
    d=\binom{\alpha}{>w}_p,
\end{equation} 
and the yield parameter is given by,
\begin{equation}
   \gamma = \frac{\log(n/k)}{\log(d)} = \frac{\log \pqty{\binom{3 \alpha}{>w}_p} - \log \pqty{\binom{3 \alpha}{\leq w}_p}}{\log \pqty{\binom{\alpha}{>w}_p)}}.
\end{equation}

Let us now study the limit of this expression when $\alpha \to \infty$, keeping the ratio $\frac{w}{(p-1)m} = t$ held fixed. The parameter $t$ ranges from $0$ to $1$. The answers can be expressed in terms of ${H}_p(\theta) = \lim_{m \to \infty} \binom{m}{\theta m}_p$ as shown in Section \ref{sec:combinatorial-tools} and Appendix \ref{asymptotic-appendix}, and is 
\begin{equation}
    \gamma_0(\theta) = \lim_{m \to \infty} \frac{\log_p \binom{m}{>w}_p - \log_p \binom{m}{\leq w}_p}{\log_p \binom{m/3}{>w}_p}  \to \begin{cases} 3 ({1-H_p(\theta)}) & \theta< \frac{p-1}{6} \\  3 \frac{1-H_p(\theta)}{H_p(3\theta)} & \theta>\frac{p-1}{6}.
    \end{cases} \label{gamma-equation-2}
\end{equation}
where $\theta=t(p-1)$. We extremize $\gamma(t)$ as a function of $t$ to find the optimal yield parameter.

Table \ref{tab:asymptotic-yield} gives numerical values of asymptotic $\gamma_0(p)$ with $p$ and their corresponding $t_0=\theta_0(p)/(p-1)$ that minimizes $\gamma_p(\theta)$. A sample computation is illustrated in Appendix \ref{asymptotic-appendix}. We also show in Appendix \ref{asymptotic-appendix} that $\gamma_0(p) \sim \frac{1}{\ln p}$ for large $p$.
\begin{table}[!htb]
\begin{minipage}[t]{0.3\textwidth}\vspace{0pt}%
    \centering
    \begin{tabular}{|c|c|c|}
        \hline
        \( p \) & \( \gamma_0(p) \) & \( t_0(p) \) \\
        \hline
         2  & 0.67799 & 0.27063 \\
3  & 0.63215 & 0.27369 \\
5  & 0.55914 & 0.27868 \\
7  & 0.50786 & 0.28230 \\
11 & 0.44108 & 0.28720 \\
13 & 0.41785 & 0.28896 \\
17 & 0.38273 & 0.29169 \\
19 & 0.36901 & 0.29278 \\
23 & 0.34663 & 0.29459 \\
        \hline
    \end{tabular}
    \caption{Asymptotic values of the yield parameter $\gamma_0(p)$.}
    \label{tab:asymptotic-yield}
\end{minipage}
\hfill
\begin{minipage}[t]{0.65\textwidth}\vspace{0pt}%
    \begin{tabular}{|c|c|c|c|c|}
        \hline
        $p$ & $m$ & $w$ & \textbf{Code Parameters} & \textbf{Length} \\
        \hline
        2 & 58 & 14 & $[[288215893050995568, 14483100716176, 21700]]_2$ & $\sim 2^{58}$ \\
        3 & 32 & 16 & $[[1852445880782154, 574308069687, 3510]]_3$ & $\sim 2^{51}$ \\
        5 & 16 & 16 & $[[152186472515, 401418110, 429]]_5$ & $\sim 2^{37}$ \\
        7 & 13 & 20 & $[[96448935471, 440074936, 231]]_7$ & $\sim 2^{36}$ \\
        11 & 7 & 19 & $[[18874416, 612755, 35]]_{11}$ & $\sim 2^{24}$ \\
        13 & 7 & 23 & $[[60848853, 1899664, 35]]_{13}$ & $\sim 2^{26}$ \\
        17 & 4 & 17 & $[[77540, 5981, 15]]_{17}$ & $\sim 2^{16}$ \\
        19 & 4 & 19 & $[[121470, 8851, 15]]_{19}$ & $\sim 2^{17}$ \\
        23 & 1 & 5 & $[[17, 6, 3]]_{23}$ & $\sim 2^{4~}$ \\
        \hline
    \end{tabular}
    \caption{Smallest punctured Reed-Muller codes with $\gamma < 1$ arising from our construction for different values of the qudit dimensionality $p$. For $p\geq 23$, punctured Reed-Solomon codes~\cite{Krishna_2019} achieve $\gamma<1$. Results for $p=2$ are from \cite{Hastings_2018}.}
    \label{tab:smallest_codes}
\end{minipage}
\end{table}

One can also ask, what are the smallest codes arising from this construction that have $\gamma<1$? These are are listed in Table \ref{tab:smallest_codes} and their block sizes are plotted as a function of $p$ in Figure \ref{fig:block-size}. We see that the required block size decreases dramatically as $p$ increases, until $p=23$, when Reed-Solomon codes achieve sublogarithmic distillation. However, in the following section, we show that one can construct punctured Reed-Muller codes that give rise to triorthogonal codes with $\gamma<1$ with much smaller block sizes, if one is willing to search for better puncture locations.  

\begin{figure}
    \centering
    \includegraphics[width=0.5\linewidth]{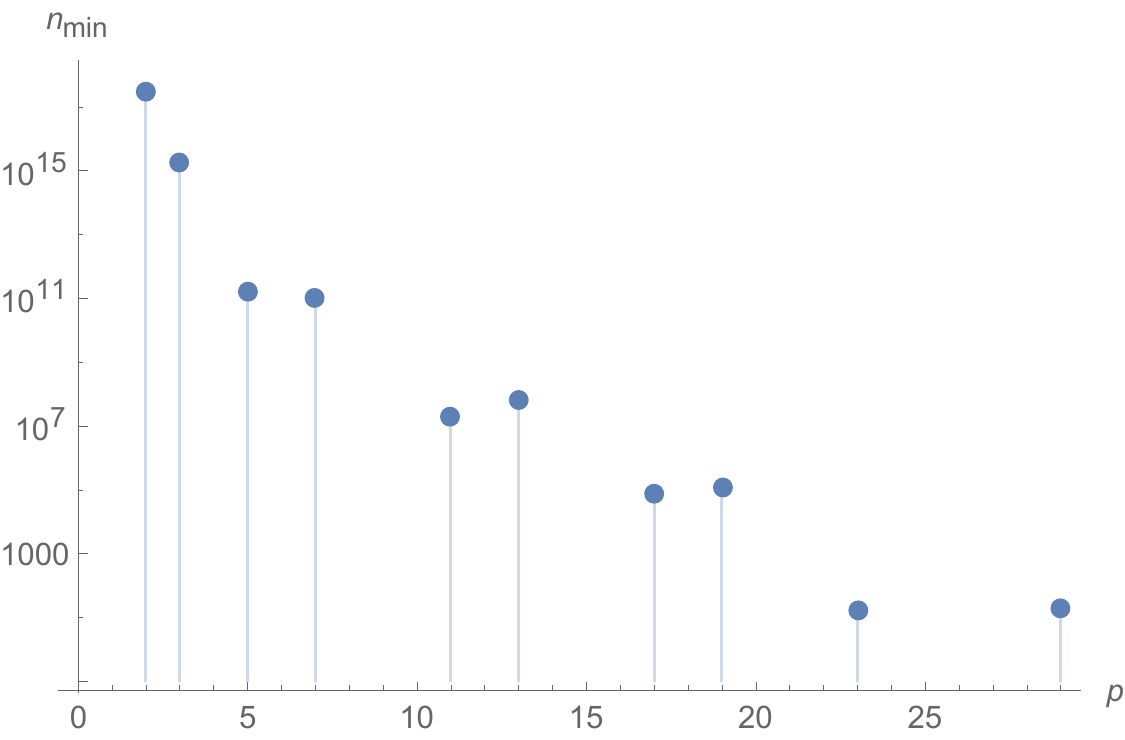}
    \caption{The minimum block size required for sublogarithmic distillation via our construction as a function of $p$. The code for $p=2$ is from \cite{Hastings_2018}, and the vertical scale is logarithmic.}
    \label{fig:block-size}
\end{figure}
\section{Sublogarithmic distillation with $\sim 500$ ququints}
\label{sec:other-schemes}

The puncturing scheme by Manhattan weight used in the earlier subsection was constructed by direct analogy from the scheme used in \cite{Hastings_2018}, and was chosen because it allowed for an explicit computation of the distance of the code for arbitrary $p$ $m$ and $r$, enabling us to compute $\gamma$ in the limit of large $n$, as in \cite{Hastings_2018}. However, there is no reason to expect that this scheme is the best way to puncture Reed-Muller codes. In this section, we attempt to ask, what are the best triorthogonal codes one can obtain from puncturing Reed-Muller codes? 

This appears to be a non-trivial question to answer analytically and computationally; here, we present the best triorthogonal codes obtained from punctured Reed-Muller codes for small $p$, that we could find from a rudimentary randomized computer search. Similar searches over binary Reed-Muller codes appeared in \cite{Haah_2018}. We could very easily find codes with values of $k$ and/or $d$ greater than or equal to what obtains from Manhattan distance puncturing. Some of the best triorthogonal codes we could find for $p=3$ and $p=5$ are listed in Table \ref{tab:punctured_reed_muller_p3}. Notably, we find a $[[519,106,5]]_5$ code with $\gamma=0.99$. To our knowledge, no triorthogonal code with $\gamma<1$ with block sizes less than 1000 for qubits or qutrits are known. We expect that randomly puncturing a ternary code Reed-Muller code with $m=7$ would give rise to a triorthogonal codes with $\gamma<1$, but confirming this was beyond the reach of our computational capabilities. The specific puncture locations that give rise to the codes in Table \ref{tab:punctured_reed_muller_p3} are in Appendix \ref{app:puncture-locations}.

\begin{table}[ht]
\centering
\begin{tabular}{|c|c|c|c|c}
\hline
\textbf{$m$} & \textbf{Code Parameters} & $A_d$ & \textbf{$\gamma$} \\
\hline
4 & $[[80, 1, 5]]_3$ & $2080$ & $2.72$ \\
4 & $[[79, 2, 4]]_3$ & $130$ & $2.65$ \\

4 & $[[72, 9, 3]]_3$ & $648$ &  $1.89$ \\ 
\hline
5 & $[[230, 13, 6]]_3$ & $572$ &  $1.60$ 
\\

5 & $[[215, 28, 5]]_3$ & $1104$ &  $1.27$ 
\\

5 & $[[206, 37, 4]]_3$ & $880$ &  $1.24$ 
\\
5 & $[[200, 43, 3]]_3$ & $1700$ &  $1.40$ \\
\hline
6 & $[[ 690, 39, 5 ]]_3$ & $1128$ &  $1.79$ \\
6 & $[[ 667, 62, 4 ]]_3$ & $3972$ &  $1.71$ \\
\hline
\end{tabular} $~~$
\begin{tabular}{|c|c|c|c|c}
\hline
\textbf{$m$} & \textbf{Code Parameters} & $A_d$ & \textbf{$\gamma$}  \\
\hline
2 & $[[24, 1, 3]]_5$ & $96$ & $2.89$ \\

2 & $[[20, 5, 2]]_5$ & $760$ & $2$ \\
\hline
3 & $[[124, 1, 4]]_5$ & $124$ & $3.48$ \\

3 & $[[112, 13, 3]]_5$ & $512$ & $1.96$ \\

\hline
4 & $[[ 519, 106, 5 ]]_5$ & $2180$ & $0.99$ \\
\hline
\end{tabular}
\caption{Parameters of some triorthogonal codes obtained from punctured Reed-Muller  spaces $RM_p(r_{\rm max}, m)$ for $p = 3$ (left) and $p=5$ (right). These codes were found by searching over randomly chosen puncture locations, and are  not claimed to be optimal. We find sublogarithmic distillation at $m=4$ for $p=5$.}
\label{tab:punctured_reed_muller_p3}
\end{table}

In the above table, we included the quantity $A_d$, which denotes the total number of weight-$d$ logical-$Z$ operators. $A_d$ can be used to estimate the distillation performance of these $[[n,k,d]]_p$ codes in practice, when only a few rounds (or even only one round) of distillation is required. As shown in \cite{Bravyi_2012, prakash2024lowoverheadqutritmagic}, for distillation with triorthogonal codes, given $n$ input qudits with depolarizing noise $\delta_{\rm in}$, we obtain $k$ output qudits, with depolarizing noise 
\begin{equation}
    \delta^{(i)}_{\rm out}=\frac{A^{(i)}_d}{(p-1)p^{d-1}} \delta_{\rm in}^d,
\end{equation} for $i=1, \ldots k$, where $A^{(i)}_d$ is the number of logical $Z_i$ operators (i.e., logical $Z$ operators on the $i$th output qudit) of weight $d$. It is slightly tedious to compute $A^{(i)}_d$ for each output qudit -- we therefore estimate the average error of an output qudit to be, following \cite{Haah_2018}, to be 
\begin{equation}
    \delta_{\rm out} \approx \frac{A_d}{\bar{n}_T(p-1)p^{d-1}} \delta_{\rm in}^d \left(1-\frac{p-1}{p} \delta_{\rm in}\right)^{n-d}
\end{equation}
where $A_d$ is the total number of weight-$d$ logical-$Z$ operators, and $\bar{n}_T \approx (1-\frac{p-1}{p} \delta_{\rm in})^{n} k$ is the average number of output magic-states. Applying this formula, we see that one round of distillation with the $[[519,106,5]]_5$ code, with $\delta_{\rm in}=10^{-3}$, we obtain $\delta_{\rm out}\approx 8\times 10^{-18}$, for a distillation cost $C=\frac{n}{\bar{n}_T} \approx 7.4$.

In our computational search, we did not seek to optimize $A_d$, only the yield parameter $\gamma$. To search for better distillation schemes, in situations where only one round of distillation is required, one should instead specify both $\delta_{\rm in}$ and $\delta_{\rm out}$, and perform a computational search for puncture locations that optimize the distillation cost, $C$, rather than the yield. 

\section{Discussion}
\label{sec:conclusion}
While there are difficulties associated with control of qudits, many have argued that there may be some advantages to working with qudits of larger dimension \cite{CampbellAnwarBrowne, campbellEnhanced, Krishna_2019} particularly with respect to fault-tolerance, and several experimental groups are working towards realizing qudit based quantum computers (e.g., \cite{Lanyon2007ManipulatingBQ, Bianchetti_2010, yurtalan2020characterization, kononenko2020characterization, Ringbauer2021AUQ, Goss_2022, Lindon2023, goss2023extending, subramanian2023efficient}). Quantifying and better understanding these potential advantages is important, as one explores all possible routes to fault-tolerant quantum computation. 

Punctured Reed-Muller codes provide a natural construction of triorthogonal codes for magic state distillation, that can be readily applied to qudits of all prime dimensions. At least for this particular family of codes, we are able to quantify the advantages associated with increasing qudit dimensionality rather unambiguously via the yield parameter, $\gamma$ that characterizes overhead of magic state distillation protocol. In particular, we find that, the puncturing scheme of \cite{Hastings_2018}, can be naturally generalized to qudits, and gives rise to asymptotic yield parameters in the limit $n\to \infty$ that decrease monotonically with increasing $p$, as shown in Table \ref{tab:asymptotic-yield}. At a more practical level, the smallest block size $n$ required to obtain a sublogarithmic overhead, characterized by $\gamma<1$, decreases rapidly with increasing $p$, from $p=2$ up to $p=23$, as shown in Table \ref{tab:smallest_codes}. These results are, of course, specific to Reed-Muller codes, and somewhat recently, better codes for magic state distillation have been proposed in the literature. Nevertheless, we believe that the results do illustrate and quantify the extent of the possible advantages for fault-tolerant overheads as one increases the qudit-dimensionality.  

In experimental realizations, one may wish to focus on smaller values of $p$. For qutrits, the advantage  is not as dramatic as for larger values of $p$, with $\sim 2^{52}$ qutrits required for sublogarithmic distillation compared to $2^{56}$ qubits \cite{Hastings_2018}, using our particular puncturing scheme. The codes in Table \ref{tab:smallest_codes} are certainly not optimal, however, and better choices of puncture locations could lead to vastly superior codes. In a randomized search, we could not find any punctured Reed-Muller code for qutrits with $\gamma<1$ for $n<729$, although our search was not exhaustive. Therefore the, simple block construction of \cite{prakash2024lowoverheadqutritmagic} still appears to be optimal for codes with less than $1000$ qutrits. For ququints, however, we could find a set of puncture locations giving rise to a triorthogonal code with parameters $[[520,105,5]]_5$ with $\gamma<1$. 

The existence of these sublogarithmic distillation protocols, as well as the $[[520,105,5]]_5$ triorthogonal code, demonstrate that the space of qudit error-correcting codes is still relatively unexplored. It would be interesting to carry out an exhaustive search for small triorthogonal codes, like that of \cite{Nezami_2022}, for qudits, however, this is computationally challenging. We expect, however, that randomized searches and educated guesses could give rise to interesting codes. A particularly interesting open question is to identify the smallest triorthogonal codes, for small values of $p$, that give rise to sublogarithmic distillation. 

\section*{Acknowledgements}
TS thanks Prof. Prasad Krishnan for fruitful discussions pertaining to the proof of the distance formula, Michael Zurel for his guidance and fruitful discussions. SP thanks Rev. Prof. Prem Saran Satsangi for inspiration and guidance. SP also thanks Prof. Nadish de Silva and Simon Fraser University for hospitality where part of this work was completed.  Both authors acknowledge the support of DST-SERB grant (CRG/2021/009137).

\appendix

\section*{Appendices}

\section{Asymptotic analysis}
\label{asymptotic-appendix}
In this section we work out asymptotic expressions for the $p$-nomial coefficients $\binom{m}{w}_p$, in the limit $m \to \infty$, with the ratio $w/m =\theta$ held fixed.  We then use these expressions to determine the yield parameter $\gamma_\infty(p)$ in the asymptotic limit $m \to \infty$. We find that, in the limit of large $p$, $\gamma_{\infty}(p) \sim \frac{1}{\ln p}$.

\subsection{Asymptotic expansions for $p$-nomial coefficients}

We can write the $p$-nomial coefficient as a contour integral, directly from its definition,
\begin{equation}
\begin{split}
    \binom{m}{w}_{p} & = \frac{1}{w!} \frac{d^w}{dx^w}\big|_{x=0}(1+x+x^2 + \ldots x^{p-1}) \\& = \frac{1}{2\pi i} \oint \frac{(1-z^p)^m}{(1-z)^m z^{w+1}} dz.
    \end{split}
\end{equation}
The integration contour may be any closed contour, counter-clockwise, enclosing the origin. Now let $w=\theta m$, and assume $m \gg 1$. Note that 
\begin{equation}
    0 \leq \theta \leq (p-1).
\end{equation}

We then have 
\begin{equation}    \begin{split}
    \binom{m}{\theta m}_{p}  = & \frac{1}{2\pi i} \oint \left(\frac{(1-z^p)}{(1-z) z^{\theta+\frac{1}{m}}}\right)^m dz. \\
    =  & \frac{1}{2\pi i} \oint e^{m f(z)} dz,
    \end{split}
\end{equation}
where, 
\begin{equation}
    f(z) = \ln \left(\frac{(1-z^p)}{(1-z) z^{\theta+1/m}}\right) \approx \ln (1-z^p)-\ln (1-z)-\theta \ln z.
\end{equation}
We evaluate this, in the limit of large $m$, by the method of steepest descent. Expressing our answers in terms of $w$ rather than $m$, the result is
\begin{equation}
    \binom{w/\theta}{w }_{p} \sim \frac{\xi ^{-w-1} \left(\frac{\xi ^p-1}{\xi -1}\right)^{\frac{ w}{\theta }}}{\sqrt{2 \pi } \sqrt{\frac{\theta +w \left(\theta +\frac{ \xi ^2}{(\xi -1)^2}-\frac{ p \left(\xi ^p+p-1\right) \xi ^p}{\left(\xi ^p-1\right)^2}\right)}{\theta  \xi ^2}}}.
\end{equation}
where $\xi(\theta)$ is the dominant solution of the saddle-point equation, $f'(z)\big|_{z=\xi}=0$, which is defined implicitly as the solution to the equation,
\begin{equation}
    \frac{(p (\xi-1)-\xi) \xi^p+\xi}{(\xi-1) \left(\xi^p-1\right)} =\theta. \label{saddle-point-equation}
\end{equation}

Defining, 
\begin{equation}
    {H}_p(\theta) = \lim_{m\to \infty} \frac{1}{m} \log_p \binom{m}{\theta m }_{p} = \lim_{w\to \infty} \frac{\theta}{w} \log_p \binom{w/\theta}{w }_{p},
\end{equation} 
we have,
\begin{equation}
    H_p(\theta) =  \log_p \left(\frac{\xi^p-1}{\xi-1}\right)-\frac{\theta}{3}\log_p (\xi). \label{asymptotic-pnomial-entropy}
\end{equation}
While this is a rather explicit expression for $H_p(\theta)$, one must solve a polynomial of degree $p-1$ to obtain $\xi(\theta)$. Let us consider a few small cases, to obtain explicit formulae.

For $p=2$, the saddle-point equation is linear, and the dominant saddle is 
\begin{equation}
    \xi(\theta)=\frac{\theta}{1-\theta},
\end{equation} 
so we find,
\begin{equation}
H_2(\theta)= -\log \left(1-{\theta }\right)-\theta  \log \left(\frac{\theta }{1-\theta }\right) .
\end{equation}
This is, of course, just the binary entropy function.

For $p=3$, the dominant saddle is 
\begin{equation}
\xi = \frac{\theta +\sqrt{3(2-\theta ) \theta +1}-1}{2(2-\theta )}.
\end{equation}
We find,
\begin{equation}
    {H}_3(\theta) = \log \left(\frac{7-3 \theta +\sqrt{1-3 \theta (\theta -2) }}{2 (\theta -2)^2}\right)-\theta  \log \left(\frac{\theta +\sqrt{1-3 \theta (\theta -2) }-1}{2 (2-\theta)}\right).
\end{equation}
A plot of $H_3(\theta)$ compared with numerical values of $\binom{m}{w}_3$ for large $m$ and $w$, is shown in Figure \ref{fig:H3}.

\begin{figure}
    \centering
    \includegraphics[width=0.5\linewidth]{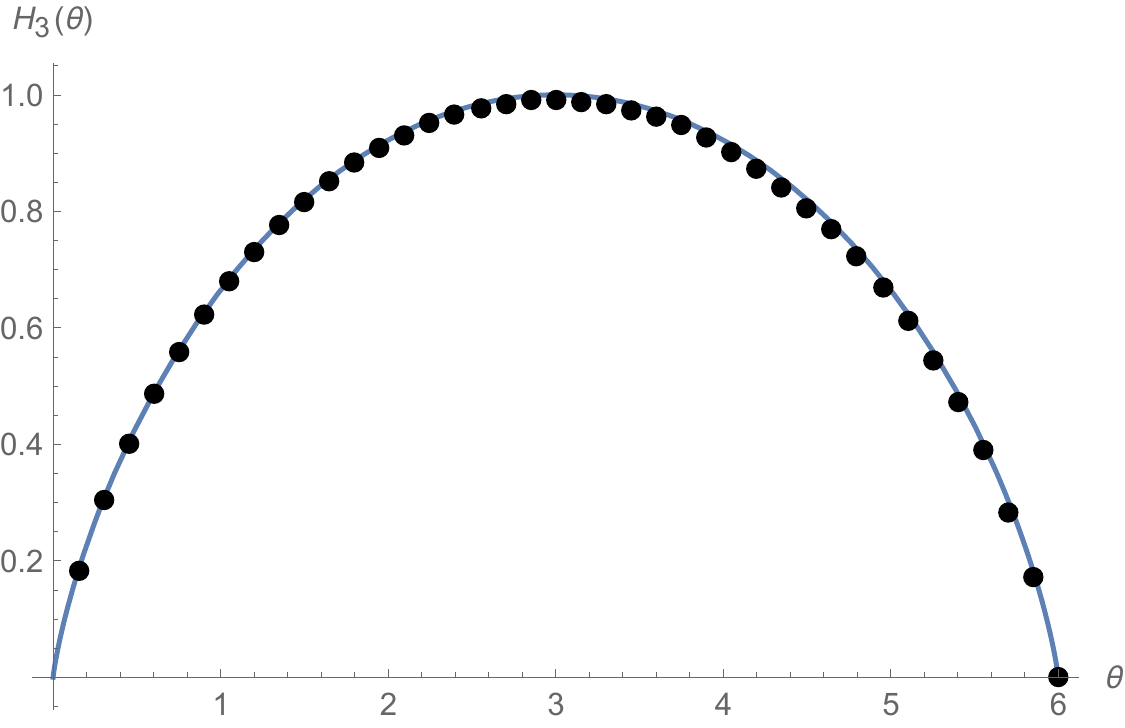}
    \caption{$H_3(2t)$  (blue) compared with numerical values (black points) computed for $m=\frac{w}{2t}$ and $w$, for $w=400$.}
    \label{fig:H3}
\end{figure}

For larger values of $p$, we find that the saddle point equation is solvable via computer algebra system, but the analytic solutions are rather unwieldy. Using these expressions, however, $H_p(\theta)$ can be computed numerically with ease. An alternative approach is to  express $H$ as a function of $\xi(\theta)$ rather than $\theta$. 

Note that the saddle-point equation determines $\theta(\xi)$, via a function that is monotonically increasing for $\xi \in (0,\infty)$. $\theta(0)=0$ and $\lim_{\xi \to \infty} \theta(\xi)=p-1$.  

We therefore have the explicit expression,
\begin{equation}
    H_p(\theta) = \log_p \left(\frac{\xi(\theta) ^p-1}{\xi(\theta) -1}\right)+\frac{ \left(\xi(\theta) ^p (\xi(\theta)(1-p) +p)-\xi(\theta) \right)}{(\xi(\theta) -1) \left(\xi(\theta) ^p-1\right)}\log_p (\xi(\theta)),
\end{equation}
where one may use $\xi(\theta) \in (0, \infty)$ instead of $\theta$ to parametrize the ratio of $w$ and $m$.  

\subsection{Asymptotic expansion of $p$-nomial coefficients in the limit of large $p$}
Let us now derive asymptotic formulae for large $p$. Let $\theta=t(p-1)$, so that $0<t<1$ for any $p$. We compute $\hat{H}_p(t)= H_p(t (p-1))$ when $p \to \infty$. The saddle point equation becomes,
\begin{equation}
    \left(z^p-1\right) (t (z-1)- z)-p (z-1) \left(t \left(z^p-1\right)- z^p\right) =0.
\end{equation}
We use the ansatz that the dominant saddle point must take the form $\xi =1 + \xi_1(t)/p + O(1/p^2)$, and will attempt to determine $\xi_1(t)$. This ansatz is consistent with numerical solutions of the saddle point equation at large $p$. We find $\xi_1(t)$ must satisfy the non-linear equation, 
\begin{equation}
t=-\frac{1}{\xi_{1}}+\frac{1}{e^{\xi_{1}}-1}+1 \label{subleading-saddle}
\end{equation}
There is a unique solution $\xi_1(t)$ for $t \in (0,1)$. In terms of $\xi_1(t)$ and $t$, we have,
\begin{equation}
\hat{H}_p(t) = \frac{t(p-1)}{w}\log_p \binom{w/t(p-1)}{w}_p  \approx 1 -\frac{t \xi_1(t)}{\ln p}+\log_p \left(\frac{e^{\xi_1(t)-1}}{\xi_1(t)}\right).
\end{equation}

To proceed further we observe, as before, that $t(\xi_1)$ is given by Equation \eqref{subleading-saddle}, which is monotonically increasing $t(-\infty) \to 0$ to $t(\infty) \to 1$, as shown in Figure \ref{fig:t-xi1}. By choosing $\xi_1 \in (-\infty, \infty)$ we can parameterize all $t \in (0,1)$. We now substitute in $t(\xi_1)$ using equation \eqref{subleading-saddle}, to obtain,
\begin{equation}
\hat{H}_p(t) = 1+\frac{1-\frac{e^{\xi_{1}(t)} \xi_{1}(t)}{e^{\xi_{1}(t)}-1}+\ln \left(\frac{e^{\xi_{1}(t)}-1}{\xi_{1}(t)}\right)}{\ln (p)}. \label{Hhat-xi1}
\end{equation}

\begin{figure}
    \centering
    \includegraphics[width=0.5\linewidth]{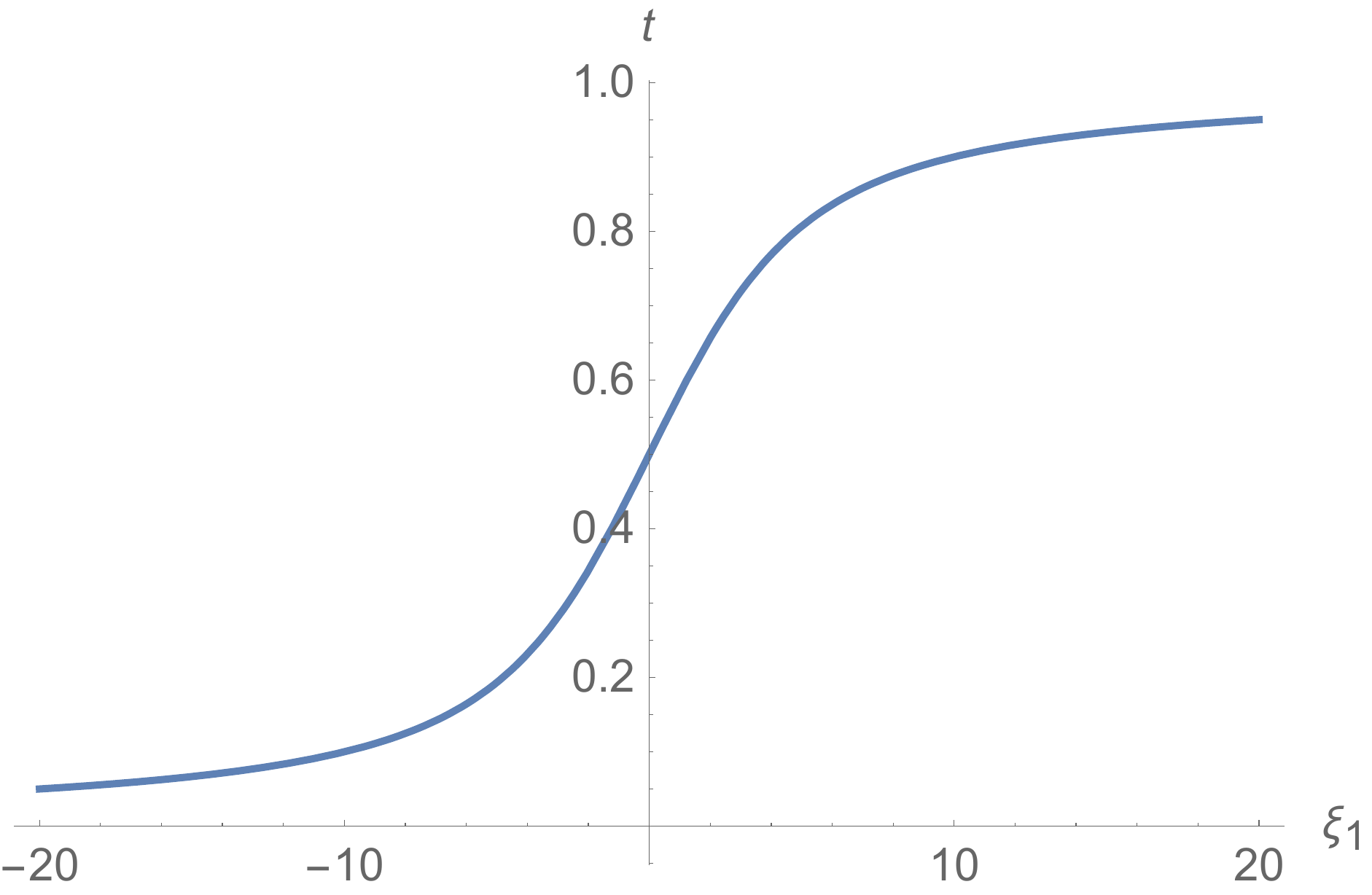}
    \caption{A plot of equation \eqref{subleading-saddle} for $t(\xi_1)$.}
    \label{fig:t-xi1}
\end{figure}

\subsection{Asymptotic limit of the yield parameter}
We use the results of the previous section to work out the asymptotic limit of the yield parameter $\gamma$ in the limit of $m \to \infty$ with $w/m=\theta$ is held fixed. Define,
\begin{equation}
H^{(\leq)}_p(\theta)=\frac{1}{m}\log \binom{m}{\leq \theta m}, \text{ and,}\quad  H^{(>)}_p(\theta)=\frac{1}{m}\log \binom{m}{> \theta m}.
\end{equation}
We clearly have 
\begin{equation}
    H^{(\leq)}_p(\theta) = \begin{cases}
        H_p(\theta)  & \theta<(p-1)/2 \\
        1  & \theta>(p-1)/2 \\
    \end{cases}
\end{equation}
and 
\begin{equation}
    H^{(>)}_p(\theta) = \begin{cases}
        1  & \theta<(p-1)/2 \\
        H_p(\theta)  & \theta>(p-1)/2 \\
    \end{cases}.
\end{equation}

The yield parameter, in the limit of large $m=3\alpha$, is given by,
\begin{equation}
    \gamma_0(\theta) = \lim_{m \to \infty} \frac{\log_p \binom{m}{>w}_p - \log_p \binom{m}{\leq w}_p}{\log_p \binom{m/3}{>w}_p}  \to \begin{cases} 3 ({1-H_p(\theta)}) & \theta< \frac{p-1}{6} \\  3 \frac{1-H_p(\theta)}{H_p(3\theta)} & \theta>\frac{p-1}{6}.
    \end{cases} \label{gamma-equation}
\end{equation}
In these expressions, we require $3\theta <  p-1$. 

For $p=3$, $\gamma_0(\theta)$ is plotted in Figure \ref{fig:gamma3} and is given by, 
\begin{equation}
\gamma_0(\theta) = \frac{3 \left(-\theta  \log (4-2 \theta )+\log \left(\frac{1}{2} \left(-3 \theta -\sqrt{1-3 (\theta -2) \theta }+7\right)\right)+\theta  \log \left(\theta +\sqrt{1-3 (\theta -2) \theta }-1\right)\right)}{3 \theta  \log (4-6 \theta )-3 \theta  \log \left(3 \theta +\sqrt{9 \theta  (2-3 \theta )+1}-1\right)+\log \left(-\frac{6}{9 \theta +\sqrt{9 \theta  (2-3 \theta )+1}-7}\right)}
\end{equation}
for $\theta>1/3$. We find the minimum is $\gamma_0=0.632$, which occurs at $\theta_0=0.547$. Optimal yield parameters for other values of $p$ are in Table \ref{tab:asymptotic-yield} and plotted in Figure \ref{fig:gamma-p-plot}.

\begin{figure}
    \centering
    \includegraphics[width=0.5\linewidth]{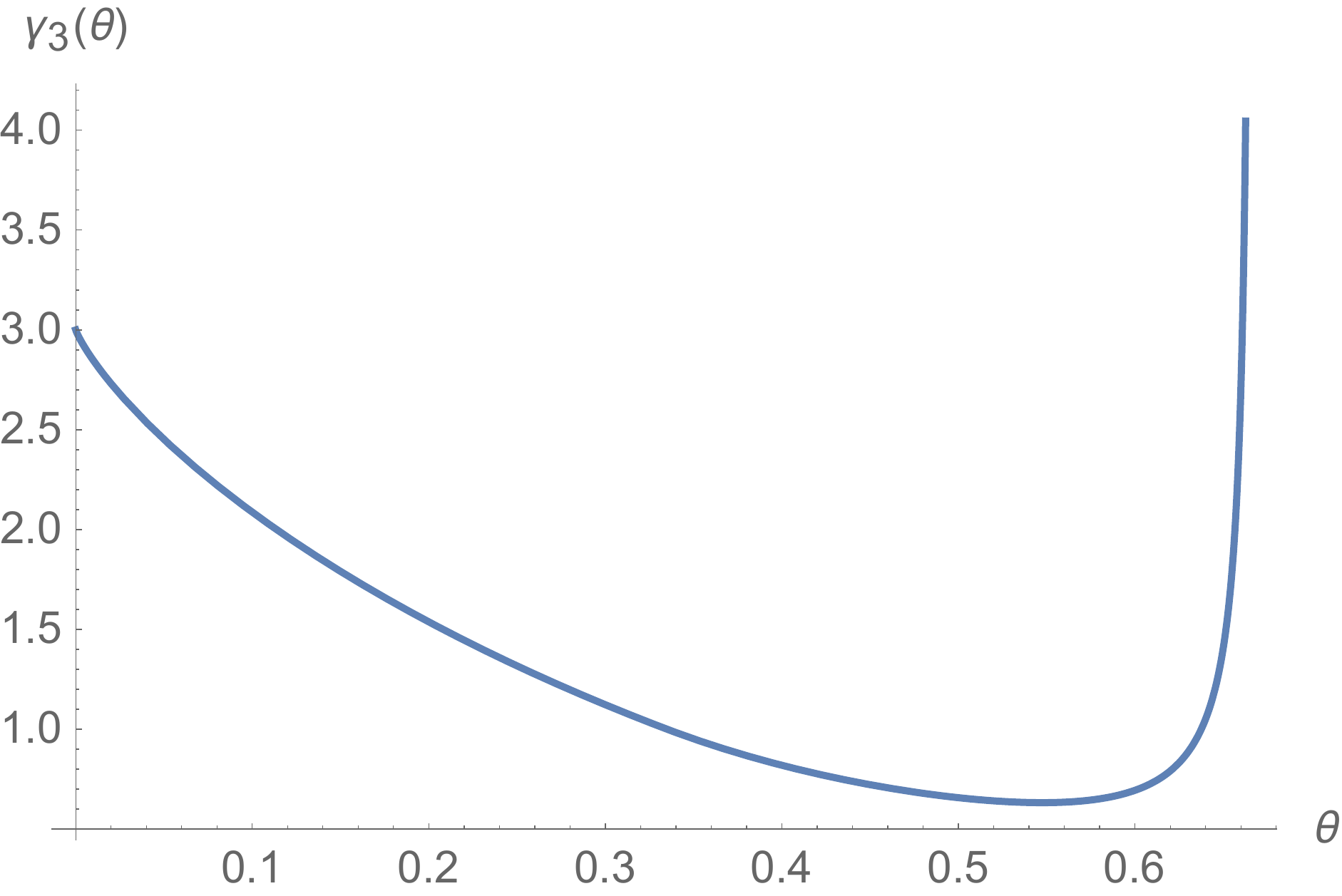}
    \caption{$\gamma(\theta)$ for $p=3$. The minimum is at $\gamma=0.632$, with $\theta=0.547$}
    \label{fig:gamma3}
\end{figure}

\begin{figure}
    \centering
    \includegraphics[width=0.5\linewidth]{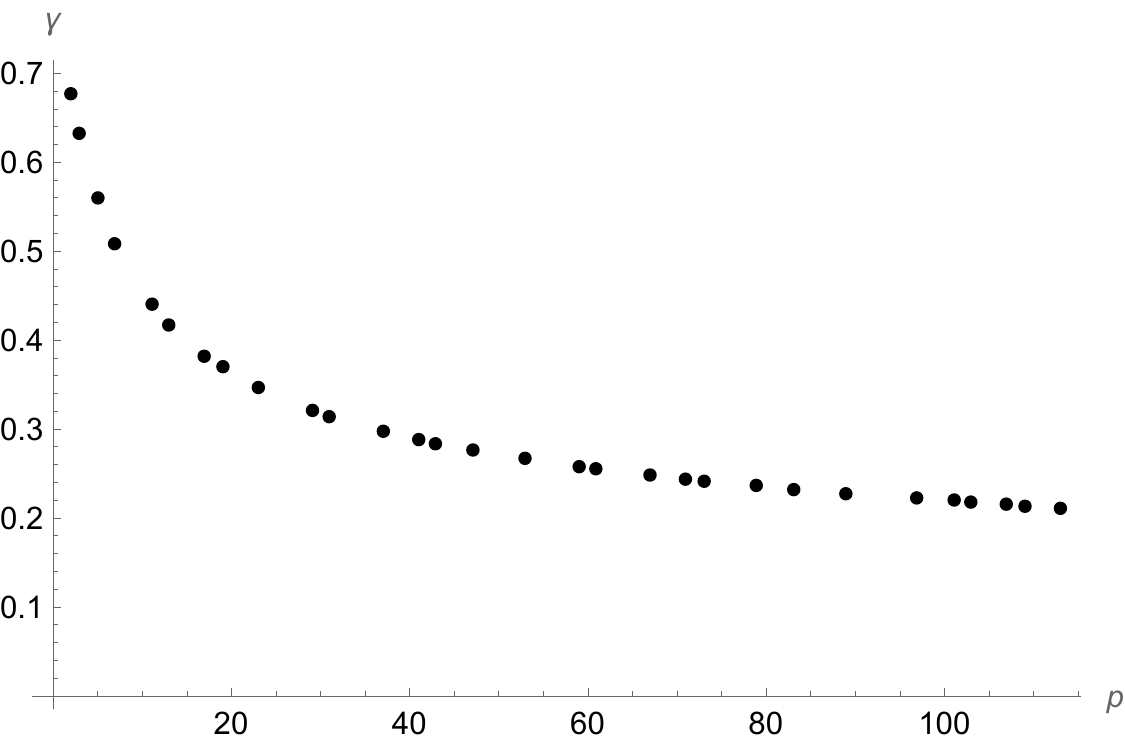}
    \caption{A plot of the best yield parameter $\gamma$ as a function of $p$.}
    \label{fig:gamma-p-plot}
\end{figure}
\subsection{Asymptotic yield in the limit of large $p$}

Let us consider the value of $\theta=(p-1)/6$, or $t=1/6$. This is not claimed to be the optimal value of $t$, but it suffices to obtain a bound. (If $t_*$ is the optimal value of $t$ that minimizes $\gamma$, the we clearly have $\gamma_p(t_*) \leq \gamma_p(1/6)$.) 
From equation \eqref{gamma-equation}, we have, 
\begin{equation}
    \gamma_p \to 3(1-\hat{H}_p(t)).
\end{equation}

Using the expressions for $\hat{H}_p(t)$ in terms of $\xi_1(t)$ in Equation \eqref{Hhat-xi1}, we have,
\begin{equation}
    \gamma_p(1/6) \to \frac{3 \frac{e^{\xi_{1}(\frac{1}{6})} \xi_{1}(\frac{1}{6})}{e^{\xi_{1}(\frac{1}{6})}-1}-3\ln \left(\frac{e^{\xi_{1}(\frac{1}{6})}-1}{\xi_{1}(\frac{1}{6})}\right)}{\ln p} \approx \frac{2.38309}{\ln p}. \label{asymptotic-formula}
\end{equation}
Note that $\xi_1(t)$ is a constant independent of $p$. Thus $\gamma_p(1/6) \sim c/\ln p$ as $p \to \infty$. We evaluated the numerical constant $c$, using $\xi_1(1/6) \approx -5.903$. Equation \eqref{asymptotic-formula} agrees well with numerics. We compare with the computation of $\gamma_p(1/6)$ exactly for the first 100 primes in Figure \ref{fig:yield-t16-plot}.

\begin{figure}
    \centering
    \includegraphics[width=0.5\linewidth]{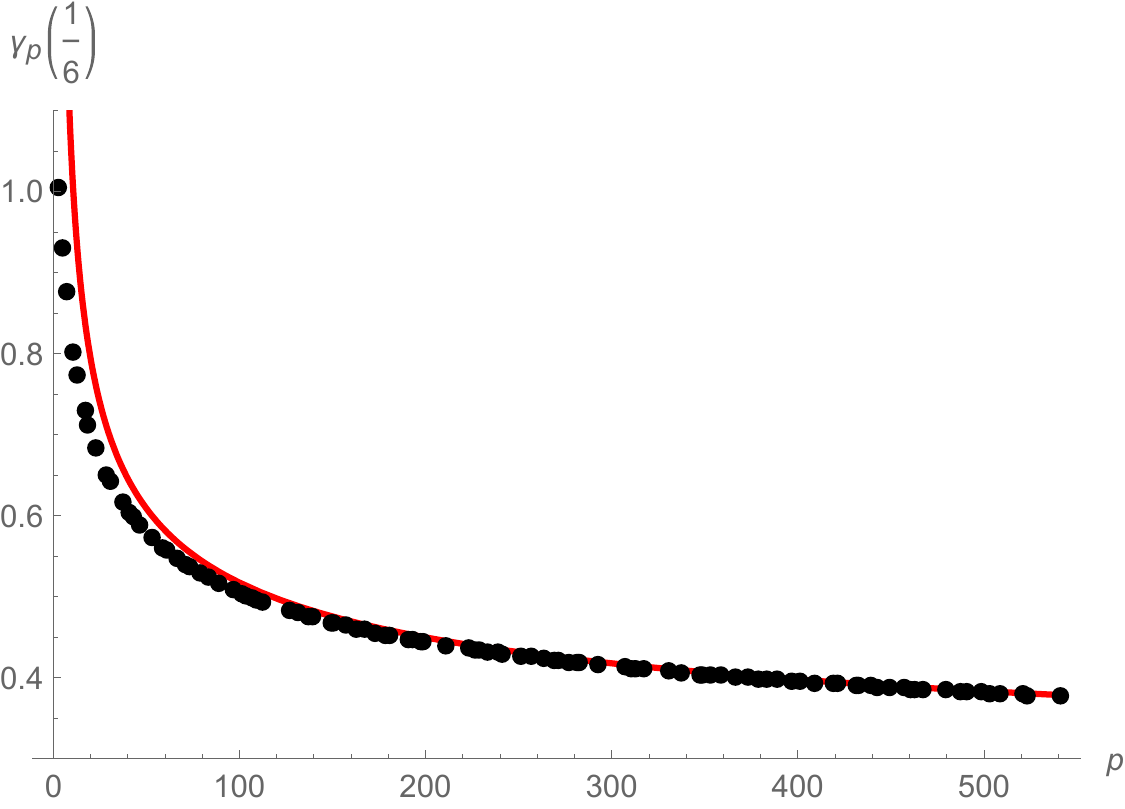}
    \caption{We plot the asymptotic yield parameter $\gamma$, evaluated at $t=\frac{w}{m(p-1)}=\frac{1}{6}$, from the asymptotic formula for large $p$, equation \eqref{asymptotic-formula} (red line), and directly in the large $m$ limit for finite $p$, (black points), for the first $100$ primes. There is good agreement at larger $p$.}
    \label{fig:yield-t16-plot}
\end{figure}

For \textit{very} large $p$, $t_* \to \frac{1}{3}$, as, for any small positive value of $\epsilon$, the denominator of $\frac{3(1-\hat{H}_p(1/3-\epsilon))}{\hat{H}_p(1-3\epsilon)} \to \frac{3(1-\hat{H}_p(1/3-\epsilon))}{1}$ for sufficiently large $p$. Thus $\gamma(t_*) \sim \frac{c}{\ln p}$, where $c \to 0.518$ from above, very slowly, as $p \to \infty$.

\section{Proof of technical lemma}
\label{app:p-nomial}
We use the following lemma in the proof of Theorem \ref{distance-theorem}.

\begin{lemma}
    \begin{equation}
    \sum_{i=0}^{p-1-A} \Delta_{p}(m, r-A, w-i) \geq \sum_{i=0}^{p-1} \Delta_{p}(m, r, w-i) = \Delta_{p}(m+1, r, w).
\end{equation}
with equality if $A=0$.
\label{distance-lemma-2}
\end{lemma}
\begin{proof}
Here $r=(p-1)\alpha+\beta$. 
\begin{enumerate}
\item[Case 1:] $\beta \geq A$.
We expand both sides and compare terms:
\begin{align}
     S_1=\sum_{i=0}^{p-1-A} \Delta_{p}(m, r-A, w-i) &= \sum_{i=0}^{p-1-A} \sum_{j=0}^{p-\beta+A-1} \binom{m - \alpha - 1}{> w - i - j}_{p}\\
    S_2 = \sum_{i=0}^{p-1} \Delta_{p}(m, r, w-i) &= \sum_{i=0}^{p-1} \sum_{j=0}^{p-\beta-1} \binom{m - \alpha - 1}{> w - i - j}_{p}  \label{s2}
\end{align}
Note that, if $A=0$, $S_1=S_2$, so assume $A>0$.

We see that,
\begin{equation}
\begin{split}
    S_1 - S_2 & =  \sum_{i=0}^{p-1-A}~\sum_{j=p-\beta}^{p-1-\beta+A} \binom{m - \alpha - 1}{> w - (i+j)}_{p} -\sum_{i=p-A}^{p-1}~\sum_{j=0}^{p-1-\beta} \binom{m - \alpha - 1}{> w - (i+j)}_{p}  \\ &= \sum_{i'=p-\beta}^{2p-\beta -A -1}~\sum_{j'=0}^{A-1} \binom{m - \alpha - 1}{> w - (i'+j')}_{p} -\sum_{i'=0}^{A-1}~\sum_{j'=p-A}^{2p-1-\beta-A} \binom{m - \alpha - 1}{> w - (i'+j')}_{p}
    \\ &= \sum_{i=p-\beta}^{p-A -1}~\sum_{j=0}^{A-1} \binom{m - \alpha - 1}{> w - (i+j)}_{p}  \geq 0,
    \end{split}
\end{equation}
and hence the theorem is proved. Here we used the fact that the summand only depends on $i+j$ in the second equality. 

\item[Case 2:]  $\beta<A$. 
\begin{align}
    S_1 = \sum_{i=0}^{p-1-A} \Delta_{p}(m, r-A, w-i) &= \sum_{i=0}^{p-1-A} \sum_{j=0}^{A-\beta} \binom{m-\alpha}{>w-i-j}_p \\
    &= \sum_{k=0}^{p-1} \sum_{i=0}^{p-1-A} \sum_{j=0}^{A-\beta} \binom{m-\alpha-1}{>w-i-j-k}_p.
\end{align}
$S_2$ is again given by eqn. \eqref{s2}.
Notice the $k$ summation in the $S_1$ is the same as the $i$ summation on the $S_2$. It only remains to prove that $$\sum_{i=0}^{p-1-A} \sum_{j=0}^{A-\beta} \binom{m-\alpha-1}{>w-i-j-k}_p \geq \sum_{j=0}^{p-1-\beta} \binom{m-\alpha-1}{>w-j-k}_p \; \forall k.$$ It is easy to see that this is true -- for any term corresponding to $j$ on the $RHS$, we always have a corresponding term on the $LHS$ that matches, and all terms are positive.
\end{enumerate}

\end{proof}

\section{Puncture locations}
\label{app:puncture-locations}
In this appendix, we list the puncture locations corresponding to the codes in section \ref{sec:other-schemes}. 

The unpunctured generator matrix for the triorthogonal space originating from a Reed-Muller code contains $p^m$ columns. Our puncture locations are given as column numbers, where the first column is $c=1$. Each column in the generator matrix for the triorthogonal space corresponds to a point in $\mathbb F_p^m$,  $\vec{x}=(x_1, \ldots, x_m)$ with $x_i \in \mathbb F_p$. The $c$th column represents the point corresponding to $c-1$ expressed in base-$p$, i.e., $c-1=x_1 + x_2 p + x_3p^2 + \ldots$. 

We only list puncture locations for codes with more than 2 punctures. For $p=3$, we have:
\begin{itemize}
    \item $[[72,9,3]]_3$:
    \begin{dmath}
        \parbox{0.8\linewidth}{ \{12, 29, 34, 36, 53, 57, 63, 67, 75\} }
    \end{dmath}
        \item $[[230,13,6]]_3$:
    \begin{dmath}
        \parbox{0.8\linewidth}{ \{43, 51, 74, 110, 140, 146, 147, 153, 180, 190, 198, 200, 228\} }
    \end{dmath}
        \item $[[215,28,5]]_3$:
    \begin{dmath}
        \parbox{0.8\linewidth}{ \{4, 11, 38, 41, 43, 51, 52, 54, 57, 59, 67, 74, 77, 88, 110, 140, 146, 147, 153, 180, 183, 190, 191, 198, 200, 206, 222, 228 \} }
    \end{dmath}
        \item $[[206,37,4]]_3$:
    \begin{dmath}
        \parbox{0.8\linewidth}{ \{5, 10, 13, 15, 23, 25, 37, 51, 55, 58, 62, 67, 78, 80, 85, 86,
93, 94, 95, 110, 113, 123, 137, 142, 146, 155, 166, 169, 180, 181, 182, 199, 
203, 207, 208, 223, 228\} }
    \end{dmath}
      \item $[[200,43,3]]_3$:
    \begin{dmath}
        \parbox{0.8\linewidth}{ \{5, 10, 13, 15, 23, 25, 27, 37, 51, 55, 58, 62, 67, 78, 80, 85,
86, 93, 94, 95, 110, 113, 114, 123, 137, 142, 146, 147, 155, 166, 169, 180, 181,
182, 199, 201, 203, 207, 208, 215, 218, 223, 228\}}
    \end{dmath}
     \item $[[ 690, 39, 5 ]]_3$:
    \begin{dmath}
        \parbox{0.8\linewidth}{ \{32, 35, 46, 71, 75, 92, 117, 144, 146, 170, 181, 198, 205, 247, 255, 256, 261, 266, 277, 289, 303, 305, 330, 
361, 372, 417, 424, 442, 477, 515, 527, 535, 557, 593, 605, 679, 713, 718, 729\} }
    \end{dmath}
     \item $[[667,62,4]]_3$:
    \begin{dmath}
       \parbox{0.8\linewidth}{  \{9, 18, 32, 35, 46, 71, 75, 85, 92, 93, 117, 144, 146, 153, 
166, 170, 174, 181, 198, 205, 238, 239, 247, 255, 256, 261, 266, 277, 289, 297, 
303, 305, 312, 330, 347, 361, 371, 372, 417, 424, 442, 445, 477, 494, 515, 521, 
522, 527, 535, 544, 557, 585, 593, 605, 648, 654, 656, 679, 713, 715, 718, 729\} }
    \end{dmath}
\end{itemize}
For $p=5$
\begin{itemize}
        \item $[[20,5,2]]_5$:
    \begin{dmath}
        \parbox{0.8\linewidth}{ \{3, 6, 12, 16, 23 \} }
    \end{dmath}
        \item $[[112,13,3]]_5$:
    \begin{dmath}
        \parbox{0.8\linewidth}{ \{13, 18, 29, 33, 34, 46, 47, 58, 61, 79, 91, 111, 124\} }
    \end{dmath}
        \item $[[519,106,5]]_5$:
    \begin{dmath}
        \parbox{0.8\linewidth}{ \{9, 16, 19, 21, 44, 51, 54, 94, 99, 101, 106, 111, 117, 118, 
119, 127, 129, 137, 143, 144, 145, 146, 147, 148, 156, 169, 183, 197, 201, 204, 
209, 220, 223, 234, 241, 256, 268, 270, 273, 275, 280, 282, 285, 293, 294, 296, 
314, 321, 322, 325, 327, 341, 344, 353, 354, 357, 369, 376, 380, 408, 413, 423, 
443, 447, 448, 451, 452, 462, 472, 474, 476, 480, 482, 490, 493, 494, 498, 504, 
511, 512, 514, 519, 520, 524, 526, 529, 530, 532, 538, 545, 549, 552, 556, 560, 
562, 568, 570, 577, 590, 595, 597, 601, 606, 609, 611, 617\} }
    \end{dmath}
\end{itemize}
\bibliographystyle{plain}
\bibliography{ref}

\end{document}